%% file: FTSpan12ver2.tex
\documentclass[11pt]{article}
\RequirePackage{algorithm,algorithmic,amssymb,epsf,epic,url}

\input{latexmac}

\denseformat

\newcommand {\ignore} [1] {}

\begin{document}

\title{Fault-Tolerant Spanners for Doubling Metrics: Better and Simpler}
\author{
Shay Solomon \thanks{Department of Computer Science and Applied Mathematics, The Weizmann Institute of Science, Rehovot 76100, Israel.
E-mail: {\tt shay.solomon@weizmann.ac.il}.
Part of this work was done while this author was a graduate student in the Department of Computer Science, Ben-Gurion University of the Negev,
under the support of the Clore Fellowship grant No.\ 81265410, the BSF grant No.\ 2008430, and the ISF grant No.\ 87209011.}}

\date{\empty}

\begin{titlepage}
\def\thepage{}
\maketitle

\begin{abstract}
In STOC'95 Arya et al.\ \cite{ADMSS95} conjectured that for any constant dimensional $n$-point Euclidean space, 
a $(1+\eps)$-spanner with constant degree, diameter $O(\log n)$ and weight $O(\log n) \cdot \omega(MST)$
can be built in $O(n \cdot \log n)$ time. Recently Elkin and Solomon \cite{ES12} (technical report, April 2012) proved this 
conjecture of Arya et al.\ in the affirmative. In fact, the proof of \cite{ES12} is more general in two ways.
First, it applies to arbitrary doubling metrics. Second, it provides a complete tradeoff between the three involved parameters
that is tight (up to constant factors) in the entire range.

Subsequently,  Chan et al.\ \cite{CLN12} (technical report, July 2012]) provided another 
proof for Arya et al.'s conjecture, which is simpler  than the proof of Elkin and Solomon \cite{ES12}.
Moreover, Chan et al.\ \cite{CLN12} also showed that one can build a fault-tolerant (FT) spanner with similar properties. 
Specifically, they showed that there exists a $k$-FT $(1+\eps)$-spanner with degree $O(k^2)$,
diameter $O(\log n)$ and weight $O(k^3 \cdot \log n) \cdot \omega(MST)$.
The running time of the construction of \cite{CLN12} was not analyzed.  

In this work we improve the results of Chan et al.\ \cite{CLN12}, using a simpler proof.
Specifically, we present a simple proof which shows that a $k$-FT $(1+\eps)$-spanner with degree $O(k^2)$,
diameter $O(\log n)$ and weight $O(k^2 \cdot \log n) \cdot \omega(MST)$ can be built in $O(n \cdot (\log n + k^2))$ time.
Similarly to the constructions of \cite{ES12} and \cite{CLN12}, our construction applies to arbitrary doubling metrics.
However, in contrast to the construction of Elkin and Solomon \cite{ES12},
our construction fails to provide a complete (and tight) tradeoff between the three involved parameters. The construction of Chan et al.\ \cite{CLN12} has this drawback too.

For random point sets in $\mathbb R^d$, we ``shave'' a factor of $\log n$ from the weight bound.
Specifically, in this case our construction provides within the same time $O(n \cdot (\log n + k^2))$, a $k$-FT $(1+\eps)$-spanner with degree $O(k^2)$,
diameter $O(\log n)$ and weight that is with high probability $O(k^2) \cdot \omega(MST)$.
\end{abstract} 
\end{titlepage}

\pagenumbering {arabic} 

\section{Introduction}
Consider a set $P$ of $n$ points in $\mathbb R^d$ and a number $t \ge 1$,
and let $G = (P,E)$ be a graph in which the weight $\omega(x,y)$ of each edge $e =(x,y)\in E$ 
is equal to the Euclidean distance $\|x-y\|$ between $x$ and $y$.   
The graph $G$ is called a \emph{$t$-spanner} for $P$ if for every $p,q \in P$,
there is a path in $G$ between $p$ and $q$ whose weight
(i.e., the sum of all edge weights in it) 
is at most $t \cdot \|p-q\|$.  
Such a path 
is called a \emph{$t$-spanner path}. The problem of constructing
Euclidean spanners 
has been studied intensively
(see, e.g., \cite{Chew86,KG92,ADDJS93,ADMSS95,NS07,DES09tech}).

Euclidean spanners find applications in geometric approximation algorithms, network topology design, distributed systems, and other areas.  
In many applications it is required to construct a
$(1+\eps)$-spanner $G = (P,E)$ that satisfies some useful
properties. First, the spanner should contain $O(n)$ (or nearly
$O(n)$) edges. Second, its {\em weight}\footnote{For convenience, we
will henceforth refer to the normalized notion of weight $\Psi(G) = {{\omega(G)}
\over {\omega(MST(P))}}$, which we call {\em lightness}.} $\omega(G) = \sum_{e
\in E}\omega(e)$ should not be much greater than the weight $\omega(MST(P))$ of
the minimum spanning tree $MST(P)$ of $P$. Third, its
\emph{(hop-)diameter} $\Lambda(G)$ should be small,  i.e., for
every   $p,q \in P$ there should be a path $\Pi$ in $G$
that contains at most $\Lambda(G)$ edges and has weight $\omega(\Pi) = \sum_{e
\in E(\Pi)} \omega(e) \le (1+\eps) \cdot \|p-q\|$. Fourth, its {\em (maximum)
degree}  $deg(G)$ should be small.

In  STOC'95   \cite{ADMSS95}, Arya et al.\ showed that for any set of $n$ points in $\mathbb R^d$ one can build
in $O(n \cdot \log n)$ time a $(1+\eps)$-spanner with constant degree,
diameter $O(\log n)$ and lightness $O(\log^2 n)$. 
They conjectured that one can obtain in the same time
a spanner with constant degree, and logarithmic diameter and lightness.

This conjecture of Arya et al.\ \cite{ADMSS95} was resolved in the affirmative recently by Elkin and Solomon \cite{ES12}.
In fact, the result of \cite{ES12} is more general in two ways. First, it applies to arbitrary \emph{doubling metrics}.\footnote{The \emph{doubling dimension} of a metric is the smallest value $d$
such that every ball $B$ in the metric can be covered by at most $2^{d}$ balls of half the radius of $B$. A metric 
is called \emph{doubling} if its doubling dimension is constant.}
Second, it provides a complete tradeoff between the   involved parameters.
Specifically, Elkin and Solomon \cite{ES12} showed that one can build in $O(n \cdot \log n)$ time a
$(1+\eps)$-spanner with degree $O(\rho)$, diameter $O(\log_\rho n + \alpha(\rho))$ and lightness $O(\rho \cdot \log_\rho n)$,
where $\rho \ge 2$ is an integer parameter and 
$\alpha$ is the inverse Ackermann function.
Due to lower bounds of Dinitz et al.\ \cite{DES09tech}, this tradeoff is tight (up to constant factors) in the entire range.

Later, Chan et al.\ \cite{CLN12} provided a simpler proof for Arya et al.'s conjecture. Moreover, they strengthened their 
construction to be \emph{fault-tolerant} (FT).\footnote{A spanner $G$ for a point set $P$ (or for any metric) is called a \emph{$k$-FT $t$-spanner}
(or simply $k$-FT spanner, if $t$ is clear from the context),
for any $t \ge 1$ and $0 \le k \le n-2$, if 
for any subset $F \subseteq P$ with $|F| \le k$, $G \setminus F$ is a $t$-spanner for $P \setminus F$.}
Specifically, they showed that there exists a $k$-FT $(1+\eps)$-spanner with degree $O(k^2)$,
diameter $O(\log n)$ and lightness $O(k^3 \cdot \log n)$.
The running time of the construction of \cite{CLN12} was not analyzed; we remark that a naive implementation requires quadratic time. 

In this work we improve the results of Chan et al.\ \cite{CLN12}, using a simpler proof.
Specifically, we present a simple proof which shows that a $k$-FT $(1+\eps)$-spanner with degree $O(k^2)$,
diameter $O(\log n)$ and lightness $O(k^2 \cdot \log n)$ can be built in $O(n \cdot (\log n + k^2))$ time, for any integer $0 \le k \le n-2$.
Similarly to the constructions of \cite{ES12} and \cite{CLN12}, our construction applies to arbitrary doubling metrics.
However, in contrast to the construction of Elkin and Solomon \cite{ES12},
our construction fails to provide a complete (and tight) tradeoff between the involved parameters. The construction of Chan et al.\ \cite{CLN12} has this drawback too.

For random point sets in $\mathbb R^d$, where $d \ge 2$ is an integer constant, we ``shave'' a factor of $\log n$ from the lightness bound.
Specifically, in this case our construction provides within the same time $O(n \cdot (\log n + k^2))$, a $k$-FT $(1+\eps)$-spanner with degree $O(k^2)$,
diameter $O(\log n)$ and lightness that is with high probability (shortly, w.h.p.) $O(k^2)$.
(It is assumed throughout the paper that $\eps > 0 $ is an arbitrarily small constant.)
\vspace{-0.02in}
\\
{\bf 1.1 ~Our and Previous Techniques.~}
The starting point of the construction of Elkin and Solomon \cite{ES12} is to ``shortcut''
an MST-tour (or any Hamiltonian path with constant lightness) of the metric,
using a 1-spanner construction from \cite{SE10} that has constant degree and logarithmic diameter and lightness.
Chan et al.\ \cite{CLN12} suggested a similar idea, which involves shortcutting the net-tree of \cite{GGN04,CGMZ05} (instead of the MST-tour)
using a generalized 1-spanner construction for tree metrics from the same paper \cite{SE10}. 
Each of the two approaches has advantages and disadvantages over the other one. 
The MST-tour has constant lightness, but it may blow up distances between points. On the other hand, distances between points
in the net-tree can be controlled, but the lightness of the net-tree may be logarithmic (as in 1-dimensional Euclidean metrics
where the points are uniformly spaced).
In both approaches the main difficulty is to bound the degree of the spanner. The argument of \cite{ES12} is more involved,
mainly because one has to overcome the obstacle that distances between points in the MST-tour may blow up.
On the bright side, since the MST-tour has constant lightness, it is possible to get a complete (and tight) tradeoff 
between all parameters
in the entire range. The argument of \cite{CLN12} is simpler due to the ``nice'' structure of the net-tree. On the negative side, 
since the lightness of the net-tree may be logarithmic, extending the basic tradeoff (constant degree, logarithmic diameter and lightness)
to the entire range is doomed to failure.

We follow the ideas of \cite{ES12} and \cite{CLN12}, but propose a   simpler argument. Instead of shortcutting
the MST-tour or the net-tree, we shortcut the underlying tree of the bounded degree spanner  of Gottlieb and Roditty \cite{GR082}.
More specifically, similarly to the argument of \cite{CLN12}, we shortcut the \emph{light subtrees} of the tree--those with distance
scales less than $\frac{\Delta_{\max}}{n}$, where $\Delta_{\max}$ is the maximum inter-point distance
in the metric. Since the spanner of \cite{GR082} has bounded degree,
this  immediately  gives rise to the desired
bounds. Thus our argument bypasses the main difficulties that arose in \cite{ES12} and \cite{CLN12}  on the way to reducing the degree.
Another advantage of our argument is that it easily and naturally extends to provide a FT spanner. On the other hand, 
the argument of Chan et al.\ \cite{CLN12} for extending their basic  spanner  construction to  provide  a FT one is  quite elaborate and non-trivial,
and is based on the ICALP'12 paper \cite{CLN12a} of these authors.

\vspace{-0.01in}
\section{The Basic Spanner Construction}  \label{basic}
\vspace{-0.01in}
In this section we provide a simple proof for the conjecture of Arya et al.\ \cite{ADMSS95}.

Let $M = (P,\delta)$ be an $n$-point doubling metric. 
A $(1+\eps)$-spanner $H$ for $M$
is called a \emph{tree-like spanner}, if it ``contains'' a tree $T$
(referred to as a \emph{tree-skeleton} of $H$)
that satisfies the following conditions:
\begin{enumerate}
\vspace{-0.03in}
\item Each vertex $v$ of $T$ is assigned a representative point $r(v) \in P$.
\\There is a 1-1 correspondence between the points of $P$ and the representatives of the leaves of $T$.
\\Each internal vertex is assigned a unique representative.
(Thus, each point of $P$
will be the representative of at most two vertices of $T$.  
In particular, there are at most $2n$ vertices in $T$.)
\vspace{-0.1in}
\item Each vertex $v$ of $T$ has a \emph{radius} $rad(v)$. The radius of vertices decreases geometrically with the level;
the radius of the root of $T$ is roughly  the maximum inter-point distance $\Delta_{\max}$,
and the radius of leaves is roughly the minimum inter-point distance $\Delta_{\min}$.
Therefore, the depth of $T$ is $O(\log \Delta)$, where $\Delta = \frac{\Delta_{\max}}{\Delta_{\min}}$ is the aspect ratio of $M$.
For any vertex $v$ in $T$, the metric distance between its representative $r(v)$
and the representative of any vertex that belongs to the subtree $T_v$ of $T$ rooted at $v$ is $O(rad(v))$;
hence the metric distance between every pair of representatives from $T_v$ is  $O(rad(v))$.
In particular, the weight of all edges between a vertex $v$ and its children in $T$ is    $O(rad(v))$;
hence the tree distance between every pair of vertices from $T_v$ is  $O(rad(v))$.
(The weight of an edge $(x,y)$ is given by the metric distance  $\delta(r(x),r(y))$ between the respective representatives.)
\vspace{-0.1in}
\item For any two points $p,q \in P$, there is a $(1+\eps)$-spanner path in $H$ between $p$ and $q$
that is composed of three consecutive parts: (a) a path ascending the edges of $T$, from the leaf $u$ whose representative is $p$
to some ancestor $u'$ of $u$ in $T$;
(b) a single \emph{lateral} edge $(u',v')$; (c) a path descending the edges of $T$, from $v'$ to the leaf $v$
whose representative is $q$. (Each edge $(u,v)$ between a pair of vertices of $T$ is translated into an edge $(r(u),r(v))$ in $H$
of weight $\delta(r(u),r(v))$.)
The weight of the ascending (respectively, descending) path is at most $O(rad(u'))$ (resp., $O(rad(v'))$).
The weight $\delta(r(u'),r(v'))$ of the lateral edge is  $\Omega(\frac{1}{\eps}) \cdot (rad(u')+rad(v'))$,
hence it \emph{dwarfs} the weights of the ascending and descending paths
(i.e., it is larger than them by  a factor of $\Omega(\frac{1}{\eps})$).  
\vspace{-0.01in}
\end{enumerate}

Gottlieb and Roditty \cite{GR082} proved the following theorem.  
\vspace{-0.01in}
\begin{theorem} [\cite{GR082}] \label{GR}
For any $n$-point doubling metric $M = (P,\delta)$, and any constant $\eps > 0$, one can
build in $O(n \cdot \log n)$ time a $(1+\eps)$-spanner $H$ 
and a tree-skeleton $T$ for $H$ (as defined above), with $deg(H) = O(1)$.   
\end{theorem}
\vspace{-0.01in}

The tree $T$ and the corresponding tree-like spanner $H$ of Gottlieb and Roditty \cite{GR082} are similar
to the   net-tree and the corresponding net-tree spanner \cite{GGN04,CGMZ05}.
In particular, the standard analysis of \cite{CGMZ05} (see also \cite{CLN12}) implies that
the lightness of the net-tree spanner is logarithmic. Using the same considerations it can be shown
that the tree-like spanner $H$ of \cite{GR082} has logarithmic lightness as well.

Similarly to the net-tree spanner of \cite{GGN04,CGMZ05}, the tree-like spanner of Gottlieb and Roditty \cite{GR082} may have a large diameter. 
More specifically, the diameter of the spanner is linear in the depth of the underlying tree.
To reduce the diameter, we employ the following tree-shortcutting theorem from \cite{SE10}.
\vspace{-0.01in}
\begin{theorem} [Theorem 3 in \cite{SE10}] \label{SE}
Let $T$ be an arbitrary $n$-vertex tree, and denote by $M_T$ the tree metric induced by $T$.
One can build in 
$O(n \cdot \log_\rho n)$ time, 
for any integer $\rho \ge 4$, a 1-spanner  for $M_T$ with $O(n)$ edges,
degree at most $deg(T) + 2\rho$ and diameter $O(\log_\rho n + \alpha(\rho))$.
\end{theorem}

Next, we describe a simple spanner construction $H^*$, which coincides with Arya et al.'s conjecture.

We start by building the spanner $H$ and its tree-skeleton $T$ that are guaranteed by Theorem \ref{GR}.
Note that $T$ contains at most $2n = O(n)$ vertices. Next, following an idea of \cite{CLN12},
we shortcut all the \emph{light subtrees} of $T$, i.e., the subtrees of $T$ with distance scales less than $\frac{\Delta_{\max}}{n}$.
More specifically, denote the light subtrees of $T$ by $T_1,\ldots,T_\ell$.
For every subtree $T_i$, we 
employ Theorem \ref{SE} in the particular case $\rho = O(1)$
to build a 1-spanner $G_i$ for the tree metric $M_{T_i}$ induced by $T_i$ with bounded degree and logarithmic diameter.
Notice that the edge weights of $G_i$ 
are assigned according to the distance function of $M_{T_i}$.
The 1-spanner $G_i$ is then converted into a graph $G^*_i$
over the point set $P$, by replacing
each edge $(u,v)$ of $G_i$ (where $u$ and $v$ are vertices of $T_i$)
with the edge $(r(u),r(v))$ between their respective representatives.
Finally, let $H^*$ be the spanner obtained from 
the union of the $\ell+1$ graphs $H, G^*_1,\ldots,G^*_\ell$.

Denote by $n_i$ the number of vertices in the subtree $T_i$, for each $1 \le i \le \ell$.
We have $\sum_{i=1}^\ell n_i \le 2n$.

By Theorems \ref{GR} and \ref{SE}, the time needed to build $H^*$ is $O(n \cdot \log n) + \sum_{i=1}^\ell O(n_i \cdot \log n_i)=O(n \cdot \log n)$.

By Theorem \ref{GR}, $deg(T) \le deg(H)  = O(1)$. Also, theorem \ref{SE} yields $deg(G_i) \le deg(T_i) + 2\rho
\le deg(T) + O(1) = O(1)$, for each $1 \le i \le \ell$. Since each point of $P$ is assigned as the representative of at most two vertices of $T$,
 we have $deg(H^*) \le deg(H) + 2 \cdot \max\{deg(G_i) ~\vert~ 1 \le i \le \ell\} = O(1)$.

The lightness of $H$ is $O(\log n)$. 
For each subtree $T_i$ and any pair $u,v$ of vertices in $T_i$, the metric distance between $r(u)$ and $r(v)$
is at most linear in the distance scale $O(\frac{\Delta_{\max}}{n})$ of $T_i$. Thus the weight of each edge of $G^*_i$ 
is $O(\frac{\Delta_{\max}}{n})$. The total weight of all graphs $G^*_1, \ldots,G^*_\ell$
(there are overall $O(n)$ edges in these graphs) is thus $O(n \cdot (\frac{\Delta_{\max}}{n})) = O(\Delta_{\max}) = O(\omega(MST(M)))$. Hence the lightness of $H^*$ is $O(\log n)$.

Finally, we show that $H^*$ is a $(1+\eps)$-spanner for $M$ with diameter $O(\log n)$.
Consider any pair $p,q \in P$ of points, and let $u,v$ be their
respective
leaves in $T$.  
Since $T$ is a tree-skeleton of $H$, there is a $(1+\eps)$-spanner
path  in $H$ between $p$ and $q$ that is composed of three consecutive parts: (a) a path ascending the edges of $T$ from $u$ to 
 $u'$; (b) a lateral edge $(u',v')$; (c) a path  descending the edges of $T$ from $v'$ to $v$.
Let $T_i$ and $T_j$ be the light subtrees to which $u$ and $v$ belong, respectively. (It is possible that $T_i = T_j$.)
Let $\tilde u$ (respectively, $\tilde v$) be the last vertex that belongs to $T_i$ (resp., $T_j$) on the path in $T$ from $u$ to $u'$ (resp., $v$ to $v'$). (It is possible that $\tilde u = u$ and/or $\tilde u = u'$. Similarly, it is possible that $\tilde v = v$ and/or $\tilde v = v'$.) 
We use the graph $G^*_i$ (respectively, $G^*_j$) to shortcut the path from $u$ to $\tilde u$
(resp., $v$ to $\tilde v$); the resulting path has the same (or smaller) weight, but   only $O(\log n)$ edges.
Also, the path from $\tilde u$ to $u'$ (respectively, $\tilde v$ to $v'$) has only $O(\log n)$ edges as well.
It follows that there is a $(1+\eps)$-spanner path between $p$ and $q$ in $H^*$ that has at most $O(\log n)$ edges.

\vspace{-0.05in}
\begin{theorem} \label{thm1}
For any $n$-point doubling metric $M = (P,\delta)$, and any constant $\eps > 0$, 
a $(1+\eps)$-spanner $H^*$ with constant degree, and diameter and lightness $O(\log n)$
can be built in $O(n \cdot \log n)$ time.
\end{theorem}
\vspace{-0.15in}
\section{The Fault-Tolerant Spanner Construction} \label{FT}
\vspace{-0.07in}
In this section we present a construction $H^*_{FT}$ of FT spanners
that is obtained by a simple modification of the basic spanner construction $H^*$ from the previous section.

To get a $k$-FT spanner, we will assign at most $2k+1$ representatives (instead of a single representative  as in the basic construction)
to any vertex $v$ of the tree-skeleton $T$. We want these representatives to be close (in terms of  metric distance) to the original representative $r(v)$ of $v$ in the basic construction.
Since the radii of vertices in $T$ decrease geometrically with the level, the original representative $r(x)$ of any descendant $x$ of $v$ in $T$ is 
  close to $r(v)$, and so it may serve as a representative of $v$ in the $k$-FT spanner.

Denote by $D(v)$ the set of all descendants of a vertex $v$ in $T$ (including $v$ itself).
Let $D^*(v)$ be a subset of $2k+1$ vertices from $D(v)$ with smallest level in $T$ (i.e., those of shortest hop-distance from $v$ in $T$);
if $|D(v)| \le 2k+1$, we have $D^*(v) = D(v)$.  
Let $R^*(v)$ be the set of   original representatives of the vertices in $D^*(v)$, i.e., $R^*(v) = \{r(x) ~\vert~ x \in D^*(v)\}$;
note that $|R^*(v)| \le |D^*(v)| \le 2k+1$.
Since each point of $P$ is assigned as the representative of at most two vertices of $T$, we have $|R^*(v)| \ge \left\lceil \frac{|D^*(v)|}{2} \right\rceil$.
Hence, either $R^*(v)$ contains the representatives of all the vertices in the subtree $T_v$ of $T$ rooted at $v$ or it must hold that $|D^*(v)| = 2k+1$,
and in the latter case we have $|R^*(v)| \ge k+1$.
Observe that each vertex in $T$ belongs to sets $D^*(w)$ of at most $2k+1$ vertices $w$ in $T$,
which implies that each point of $P$  belongs to sets $R^*(w)$ of at most $4k+2$ vertices $w$ in $T$.

The FT spanner construction $H^*_{FT}$ is obtained from the basic construction $H^*$ in the following way:
for each edge $(r(u),r(v))$ in the basic spanner construction $H^*$ (which is associated with the edge $(u,v)$ between the respective
vertices of $T$),
our FT spanner construction $H^*_{FT}$
will contain a complete bipartite graph between $R^*(u)$ and $R^*(v)$.

It is easy to see that, given the tree-skeleton $T$ and the basic spanner construction $H^*$, the FT spanner construction $H^*_{FT}$
can be built in $O(k^2 \cdot n)$ time.    
The overall running time is therefore $O(n  \cdot (\log n + k^2))$.

Consider an arbitrary point $p \in P$, and let $w$ be a vertex in $T$ such that $p \in R^*(w)$.
For every edge $(w,x)$ (associated with an edge $(r(w),r(x))$ of $H^*$) that is incident on $w$ in the basic construction,
$p$ is connected to all $O(k)$ points of $R^*(x)$ in $H^*_{FT}$.
Since $deg(H^*) = O(1)$, this contributes $O(k)$ units to the degree of $p$.
Since $p$ may belong to sets $R^*(w)$ of at most  $O(k)$ vertices $w$ in $T$, 
it follows that  the degree of $p$ in $H^*_{FT}$ is at most $O(k^2)$. Hence $deg(H^*_{FT}) = O(k^2)$.

Observe that each edge $(r(u),r(v))$ of $H^*$ is replaced by $O(k^2)$ edges of roughly the same weight in   $H^*_{FT}$.
More specifically, the weight of each of these $O(k^2)$ 
edges is greater than the weight $\delta(r(u),r(v))$ of the original edge $(r(u),r(v))$ by an additive factor of
$O(rad(u) + rad(v))$.
Since the basic construction $H^*$ has constant degree, it follows that $\omega(H^*_{FT}) \le O(k^2) \cdot (\omega(H^*) + \sum_{v \in T} rad(v))$. 
Using the standard analysis of \cite{CGMZ05,CLN12} 
we conclude that the lightness of the FT spanner construction $H^*_{FT}$ is $O(k^2 \cdot \log n)$. 

Finally, we show that for any pair $p,q \in P \setminus F$ of functioning  points (where $F \subseteq P$ is an arbitrary set of non-functioning points with $|F| \le k$), there is a $(1+O(\eps))$-spanner path
in $H^*_{FT} \setminus F$ with $O(\log n)$ edges.   
Let $u$ and $v$ be the respective leaves of $p$ and $q$ in $T$, and consider the $(1+\eps)$-spanner path $\Pi$ between $p$ and $q$ in $H^*$.
It has $O(\log n)$ edges, and it is composed of three consecutive parts: (a) a path $\Pi_u = (u = u_0,\ldots,u_i = u')$ from $u = u_0$ to 
 $u_i = u'$, such that $u_{i'}$ is an ancestor of $u_{i'-1}$ in $T$, for each $1 \le i' \le i$; 
 (b) a lateral edge $(u',v')$;
 (c) a path $\Pi_v = (v' = v_0,\ldots,v_j = v)$ from $v' = v_0$ to $v_j = v$, such that $v_{j'}$ is a descendant of $v_{j'-1}$ in $T$,
 for each $1 \le j' \le j$.
Note that $p \in R^*(u)$.
More generally, for each vertex $u_{i'}$ of $\Pi_u$, either $p \in R^*(u_{i'})$ or $|R^*(u_{i'})| \ge k+1$.
Thus $R^*(u_{i'})$ contains at least one functioning point, for each $0 \le i' \le i$.
Take $f(u_{i'})$ to be $p$ whenever possible, i.e., for every $0 \le i' \le i$ such that $p \in R^*(u_{i'})$;
otherwise take $f(u_{i'})$ to be an arbitrary functioning point from $R^*(u_{i'})$.  
(Note that $f(u) = p$.)
Similarly, $R^*(v_{j'})$ contains at least one functioning point, for each $0 \le j' \le j$.
Take $f(v_{j'})$ to be $q$ whenever possible; otherwise take $f(v_{j'})$ to be an arbitrary functioning point from $R^*(v_{j'})$.
(Note that $f(v) = q$.)
By replacing each vertex $x$ of $\Pi$ with the respective functioning point $f(x)$, we obtain a path $\Pi_{FT}$ in $H^*_{FT} \setminus F$
between $f(u) = p$ and $f(v) = q$
with the same number $O(\log n)$ of edges.
Since the path $\Pi_{FT}$ is a perturbation of the original $(1+\eps)$-spanner path $\Pi$ (and the radii of vertices decrease geometrically with the level), it incurs only a small additional cost of 
$O(rad(u') + rad(v')) = O(\eps) \cdot \delta(r(u'),r(v')) = O(\eps) \cdot \delta(p,q)$.
In other words, the stretch increases by an additive factor of $O(\eps)$, from $1+\eps$ to $1+O(\eps)$. One can   reduce the stretch back to $1+\eps$, at the expense of increasing 
the other parameters
by some function of $\eps$.
\begin{theorem} \label{thm2}
For any $n$-point doubling metric $M = (P,\delta)$, any constant $\eps > 0$, and any integer $0 \le k \le n-2$, 
a $k$-FT $(1+\eps)$-spanner $H^*_{FT}$ with degree $O(k^2)$, diameter $O(\log n)$ and lightness $O(k^2 \cdot \log n)$
can be built in $O(n \cdot (\log n + k^2))$ time.
\end{theorem}

\section{Random Point Sets}
\vspace{-0.02in}
In this section we improve the lightness bound of Theorem  \ref{thm2} for random point sets in $\mathbb R^d$.
This improvement does not apply to the case $d = 1$, where we show that the lightness must be logarithmic.
\vspace{0.12in}
\\
{\bf 4.1 ~Constant-Dimensional Random Point Sets.~}
Consider a set $V$ of $n$ points that are chosen independently and uniformly at random from the unit $d$-dimensional
cube, for an integer constant $d \ge 2$.
We will show that the net-tree spanner of \cite{GGN04,CGMZ05} for $V$ has w.h.p.\ constant lightness. Using the same considerations
it can be shown that the tree-like spanner $H$ of \cite{GR082} for $V$ has w.h.p.\ constant lightness as well.
Having reduced the lightness bound of $H$ from $O(\log n)$ to $O(1)$, we repeat the arguments of Sections \ref{basic} and \ref{FT}.
As an immediate consequence, the lightness bound of the basic spanner construction $H^*$ from Section \ref{basic} is reduced  from
$O(\log n)$ to $O(1)$, and the lightness bound of the FT spanner construction $H^*_{FT}$ from Section \ref{FT}
is reduced  from $O(k^2 \cdot \log n)$ to $O(k^2)$.

We restrict attention to the case $d = 2$. However, our argument   extends to higher constant dimensions.

The following lemma from \cite{NS07} provides a probabilistic lower bound on the weight of $MST(V)$. 
\begin{lemma} [Lemma 15.1.6 in \cite{NS07}]
For a set $V$ of $n$ points that are chosen independently and uniformly at random from 
the unit square, it holds w.h.p.\ that $\omega(MST(V)) = \Omega(\sqrt{n})$. 
\end{lemma}

The following lemma 
shows that the net-tree spanner for $V$ has w.h.p.\ constant lightness.
\begin{lemma} \label{maini}
For \emph{any} set $S$ of $n$ points in the unit square, the net-tree spanner has weight
$O(\sqrt{n})$.
\end{lemma}
\begin{proof}
In level 0 (the bottom level) of the net-tree, the nodes (called \emph{net points}) have radius $r_0 \approx \frac{\Delta_{\min}}{2}$.
More generally, in level $i = 0,\ldots,i^*$, where $i^* \approx \log \Delta$,
the net points have radius $r_i = 2^i \cdot r_0$.
Denote by $n_i$ the number of net points in level $i$. The sequence $n_0,\ldots,n_{i^*}$ is monotone non-increasing,
with $n_0 \le n$ and $n_{i^*} = 1$. Since the distance between any pair of $i$-level net points is at least $r_i$,
there is at most one such point within any circle of area smaller than $\pi \cdot \left(\frac{r_i}{2}\right)^2$.
It follows that $n_i = O(\frac{1}{r_i^2})$, for all $i = 0,\ldots,i^*$.

The standard analysis of \cite{GGN04,CGMZ05} shows that at each level, each net point has only a constant number of neighbors
at that level. The weight of all edges that are incident on a net point is greater than its radius
by at most a constant factor. Hence, it suffices to bound
the sum of radii of all net points.

Let $j$ be the last index in $\{0,\ldots,i^*\}$ that satisfies $r_{j} < \frac{1}{\sqrt{n}}$. If $r_0 \ge \frac{1}{\sqrt{n}}$, set $j = -1$.
The sum of radii of all net points in levels $i = 0,\ldots,j$ is given by 
$\sum_{i=0}^{j} n_{i} \cdot r_{i} ~\le~ \sum_{i=0}^{j} n \cdot \frac{r_{j}}{2^i}
~\le~ \sum_{i=0}^{j} \frac{n}{2^i \cdot \sqrt{n}} ~=~ O(\sqrt{n}).$
Assuming $j \le i^* - 1$, we have $r_{j+1} \ge  \frac{1}{\sqrt{n}}$.  
The sum of radii of all net points in levels $i = j+1,\ldots,i^*$ is given by 
$\sum_{i=j+1}^{i^*} n_{i} \cdot r_{i} 
~=~ \sum_{i=j+1}^{i^*} O(\frac{1}{r_i}) ~=~ \sum_{i=0}^{i^*-j-1} O(\frac{1}{2^i \cdot r_{j+1}}) ~=~ O(\sqrt{n})$.
The lemma follows.
\QED
\end{proof}
\begin{theorem} \label{thm3}
For any set $V$ of $n$ points that are chosen independently and uniformly at random from the unit $d$-dimensional
cube, where $d \ge 2$ is an integer constant,
any constant $\eps > 0$, and any integer $0 \le k \le n-2$, 
a $k$-FT $(1+\eps)$-spanner ${H}^{*}_{FT}$ with degree $O(k^2)$, diameter $O(\log n)$ and lightness w.h.p.\ $O(k^2)$
can be built in $O(n \cdot (\log n + k^2))$ time.
\end{theorem}
{\bf Remark:} The diameter of the tree-like spanner of \cite{GR082} is  $O(\log \Delta)$, where $\Delta$ is the aspect-ratio of the metric.
It is not difficult to show that the aspect ratio of random point sets in the unit $d$-dimensional cube is w.h.p.\ polynomially bounded, and so
the diameter of the tree-like spanner of \cite{GR082} is w.h.p.\ $O(\log n)$. 
Hence, for random point sets in the unit $d$-dimensional cube,
the tree-like spanner of \cite{GR082} achieves w.h.p.\ logarithmic diameter and w.h.p.\ constant lightness,
together with constant degree. (While the bounds on the diameter and lightness are probabilistic,   the bound on the degree is deterministic.)
Nevertheless, it is still advantageous to shortcut the light subtrees, because the diameter bound of the resulting spanner
becomes deterministic,
whereas all the other measures do not increase by more than constant factors.
\vspace{0.12in}
\\
{\bf 4.2 ~1-Dimensional Random Point Sets.~}
Let $S$ be an arbitrary set of $n$ evenly spaced points on the $x$-axis, and denote the distance between consecutive points by $\eta$.
Dinitz et al.\ \cite{DES08} proved that any spanner $H$ for $S$ with diameter $O(\log n)$ has lightness $\Omega(\log n)$, 
or equivalently, weight $\Omega(\log n \cdot (\eta \cdot n))$.
This lower bound on the weight was generalized
for Steiner spanners in \cite{ES09}.
We remark that the points need not be evenly spaced for this lower bound on the weight to hold.
That is, the lower bound $\Omega(\log n \cdot (\eta \cdot n))$ on the weight remains valid as long as the minimum distance
between consecutive points is at least $\eta$. \\We summarize this result in the following theorem.
\begin{theorem} [\cite{DES08,ES09}] \label{steinerLow}
Let $S$ be a set of $n$ points on the $x$-axis in which the minimum distance between consecutive points is at least $\eta$.
Any (possibly Steiner) spanner for $S$ with diameter $O(\log n)$ has weight $\Omega(\log n \cdot (\eta \cdot n))$.
\end{theorem}

Consider a set $V$ of $n$ points that are chosen independently and uniformly at random from the unit interval.
Clearly, the weight $\omega(MST(V))$ of the $MST$ for $V$ is w.h.p.\ $\Theta(1)$.

We argue that w.h.p.\ there exists a subset $\tilde V$ of $V$ that consists of $\tilde n = c \cdot n$ points, for some constant $0 < c < 1$,
such that the minimum distance between consecutive points in $\tilde V$ is at least $\frac{1}{n}$. 
The proof of this assertion
follows  similar lines as those in the proof of Lemma 15.1.6 from \cite{NS07}, and is thus omitted.

Let $H$ be a (possibly Steiner) spanner for $V$ with diameter $O(\log n) = O(\log \tilde n)$. Note that $H$ is a Steiner spanner for $\tilde V$, 
where all points of $V \setminus \tilde V$ (and possibly some other points too) serve as Steiner points. 
By Theorem \ref{steinerLow} (substituting $S$ with $\tilde V$, $n$ with $\tilde n$,   and $\eta$ with $\frac{1}{n}$), 
$\omega(H) ~=~ \Omega(\log \tilde n \cdot (\frac{1}{n} \cdot \tilde n)) ~=~ \Omega(\log n)$, and so w.h.p.\ $\omega(H) = \Omega(\log n) \cdot \omega(MST(V))$.
Consequently, the lightness $\Psi(H)$ of $H$ is w.h.p.\ $\Omega(\log n)$.
\begin{corollary}
Let $V$ be a set of $n$ points that are chosen independently and uniformly at random from the unit interval.
Any (possibly Steiner) spanner for $V$ with diameter $O(\log n)$ has lightness w.h.p.\ $\Omega(\log n)$.
\end{corollary}
\ignore{
And from this point on we can save items 5 and 6 below.
The argument will be extremely short, and I won't take credit for it.

================

4. By bringing points of V closer to each other in an appropriate way, we can transform V into V', where for 
  each 1 \le i \le n, v_i is transformed into v'_i; v'_1 < v'_2 < ... < v'_n; and the following two conditions hold:
  a) For all 1 \le i < j \le n, ||v'_i - v'_j|| \le ||v_i - v_j||.
  b) The distance between any pair of consecutive points in U' = {u'_1,u'_2,...,u'_{\tilde n}} is exactly .
}
\vspace{0.1in}
{\bf Acknowledgments.~}
The author is grateful to Michael Elkin and Michiel Smid for  helpful discussions.



\end {document}

%% file: latexmac.tex
\def\denseformat{
\setlength{\textheight}{9.5in}
\setlength{\textwidth}{6.9in}
\setlength{\evensidemargin}{-0.3in}
\setlength{\oddsidemargin}{-0.3in}
\setlength{\headsep}{10pt}
\setlength{\topmargin}{-0.44in}
\setlength{\columnsep}{0.375in}
\setlength{\itemsep}{0pt}
}

\newtheorem{theorem}{Theorem}[section]

\newtheorem{lemma}[theorem]{Lemma}

\newtheorem{corollary}[theorem]{Corollary}

\newtheorem{comment}[theorem]{Comment}

\def\boldhead#1:{\par\vskip 7pt\noindent{\bf #1:}\hskip 10pt}
\def\ithead#1:{\par\vskip 7pt\noindent{\it #1:}\hskip 10pt}

\def\inline#1:{\par\vskip 7pt\noindent{\bf #1:}\hskip 10pt}
\def\midinline#1:{\par\noindent{\bf #1:}\hskip 10pt}
\def\dnsinline#1:{\par\vskip -7pt\noindent{\bf #1:}\hskip 10pt}
\def\ddnsinline#1:{\newline{\bf #1:}\hskip 10pt}
\def\largeinline#1:{\par\vskip 7pt\noindent{\large\bf #1:}\hskip 10pt}
\long\def\comment #1\commentend{}
\long\def\commhide #1\commhideend{}
\long\def\commfull #1\commend{#1}
\long\def\commabs #1\commenda{}
\long\def\commtim #1\commendt{#1}
\long\def\commb #1\commbend{}
\long\def\commedit #1\commeditend{} 

\long\def\commB #1\commBend{}       

\long\def\commex #1\commexend{}     

\long\def\commsiena #1\commsienaend{}  
                                         
\long\def\commBI #1\commBIend{}  
                                         
\long\def\CProof #1\CQED{}

\def\blackslug{\hbox{\hskip 1pt \vrule width 4pt height 8pt
    depth 1.5pt \hskip 1pt}}
\def\QED{\quad\blackslug\lower 8.5pt\null\par}

\long\def\PPP#1{\noindent{\bf Proof:}{ #1}{\quad\blackslug\lower 8.5pt\null}}

\long\def\denspar #1\densend
{#1}

\newif\ifnotesw\noteswtrue
   {\ifnotesw\marginpar[\hfill\(\top\)]{\(\top\)}\fi}%
   {\ifnotesw\marginpar[\hfill\(\bot\)]{\(\bot\)}\fi}

\newcommand{\mnote}[1]%
    {\ifnotesw\marginpar%
        [{\scriptsize\it\begin{minipage}[t]{\marginparwidth}
        \raggedleft#1%
                        \end{minipage}}]%
        {\scriptsize\it\begin{minipage}[t]{\marginparwidth}
        \raggedright#1%
                        \end{minipage}}%
    \fi}

\def\MathF{\hbox{\rm I\kern-2pt F}}
\def\MathP{\hbox{\rm I\kern-2pt P}}
\def\MathR{\hbox{\rm I\kern-2pt R}}
\def\MathZ{\hbox{\sf Z\kern-4pt Z}}
\def\MathN{\hbox{\rm I\kern-2pt I\kern-3.1pt N}}
\def\MathC{\hbox{\rm \kern0.7pt\raise0.8pt\hbox{\footnotesize I}
\kern-4.2pt C}}
\def\MathQ{\hbox{\rm I\kern-6pt Q}}

\newsavebox{\ttop}\newsavebox{\bbot}

\def\eps{\epsilon}




%% file: FTSpan12ver2.bbl
\begin{thebibliography}{10}\setlength{\itemsep}{-1ex}\small

\bibitem{ADDJS93}
I.~Alth$\ddot{\mbox{o}}$fer, G.~Das, D.~P. Dobkin, D.~Joseph, and J.~Soares.
\newblock On sparse spanners of weighted graphs.
\newblock {\em Discrete \& Computational Geometry}, 9:81--100, 1993.

\bibitem{ADMSS95}
S.~Arya, G.~Das, D.~M. Mount, J.~S. Salowe, and M.~H.~M. Smid.
\newblock {E}uclidean spanners: short, thin, and lanky.
\newblock In {\em Proc. of 27th STOC}, pages 489--498, 1995.

\bibitem{CGMZ05}
H.~T.-H. Chan, A.~Gupta, B.~M. Maggs, and S.~Zhou.
\newblock On hierarchical routing in doubling metrics.
\newblock In {\em Proc. of 16th SODA}, pages 762--771, 2005.

\bibitem{CLN12a}
T.-H.~H. Chan, M.~Li, and L.~Ning.
\newblock Sparse fault-tolerant spanners for doubling metrics with bounded
  hop-diameter or degree.
\newblock In {\em ICALP (1)}, pages 182--193, 2012.

\bibitem{CLN12}
T.-H.~H. Chan, M.~Li, and L.~Ning.
\newblock Incubators vs zombies: Fault-tolerant, short, thin and lanky spanners
  for doubling metrics.
\newblock {\em {T}echnical Report, CoRR abs/1207.0892}, July, 2012.

\bibitem{Chew86}
L.~P. Chew.
\newblock There is a planar graph almost as good as the complete graph.
\newblock In {\em Proc. of 2nd SOCG}, pages 169--177, 1986.

\bibitem{DES08}
Y.~Dinitz, M.~Elkin, and S.~Solomon.
\newblock Shallow-low-light trees, and tight lower bounds for {E}uclidean
  spanners.
\newblock In {\em Proc. of 49th FOCS}, pages 519--528, 2008.

\bibitem{DES09tech}
Y.~Dinitz, M.~Elkin, and S.~Solomon.
\newblock Low-light trees, and tight lower bounds for {E}uclidean spanners.
\newblock {\em Discrete \& Computational Geometry}, 43(4):736--783, 2010.

\bibitem{ES09}
M.~Elkin and S.~Solomon.
\newblock Narrow-shallow-low-light trees with and without {S}teiner points.
\newblock In {\em Proc. of 17th ESA}, pages 215--226, 2009.

\bibitem{ES12}
M.~Elkin and S.~Solomon.
\newblock Optimal {E}uclidean spanners: really short, thin and lanky.
\newblock {\em {T}echnical Report CS12-04, Ben-Gurion University}, April 4,
  2012.
\newblock See also {T}echnical Report, {CoRR} abs/1207.1831, July, 2012.

\bibitem{GGN04}
J.~Gao, L.~J. Guibas, and A.~Nguyen.
\newblock Deformable spanners and applications.
\newblock In {\em Proc. of 20th SoCG}, pages 190--199, 2004.

\bibitem{GR082}
L.~Gottlieb and L.~Roditty.
\newblock An optimal dynamic spanner for doubling metric spaces.
\newblock In {\em Proc. of 16th ESA}, pages 478--489, 2008.
\newblock Another version of this paper is available via
  \url{http://cs.nyu.edu/~adi/spanner2.pdf}.

\bibitem{KG92}
J.~M. Keil and C.~A. Gutwin.
\newblock Classes of graphs which approximate the complete {E}uclidean graph.
\newblock {\em Discrete \& Computational Geometry}, 7:13--28, 1992.

\bibitem{NS07}
G.~Narasimhan and M.~Smid.
\newblock {\em Geometric Spanner Networks}.
\newblock Cambridge University Press, 2007.

\bibitem{SE10}
S.~Solomon and M.~Elkin.
\newblock Balancing degree, diameter and weight in {E}uclidean spanners.
\newblock In {\em Proc. of 18th ESA}, pages 48--59, 2010.

\end{thebibliography}
